\documentclass[preprint,12pt]{elsarticle}
\usepackage{amsmath,amsthm,amssymb,amsfonts,mathdots,graphicx,algorithm,algpseudocode}
\usepackage{a4wide}
\usepackage{bm}
\usepackage{tikz}
\usetikzlibrary{decorations.markings}

\newcommand{\defeq}{\stackrel{\text{def}}{=}}

\newtheorem{thm}{Theorem}

\newtheorem{lemma}{Lemma}[thm]
\newtheorem{cor}{Corollary}[thm]

\theoremstyle{definition}
\newtheorem{defi}{Definition}

\begin{document}

\begin{frontmatter}
\title{The IMP game: Learnability, approximability and adversarial learning beyond $\Sigma^0_1$}
\author[monash]{Michael Brand}
\ead{michael.brand@monash.edu}
\author[monash]{David L. Dowe}
\ead{david.dowe@monash.edu}
\address[monash]{Faculty of IT (Clayton), Monash University, Clayton, VIC 3800, Australia}

\begin{abstract}
We introduce a problem set-up we call the Iterated Matching Pennies (IMP) game
and show that it is a powerful framework for the study of three problems:
adversarial learnability, conventional (i.e., non-adversarial) learnability and
approximability. Using
it, we are able to derive the following theorems. (1) It is possible
to learn by example all of $\Sigma^0_1 \cup \Pi^0_1$ as well as some supersets;
(2) in adversarial learning (which we describe as a pursuit-evasion game),
the pursuer has a winning strategy (in other words, $\Sigma^0_1$ can be learned
adversarially, but $\Pi^0_1$ not); (3) some languages in $\Pi^0_1$ cannot be
approximated by any language in $\Sigma^0_1$.

We show corresponding results also for $\Sigma^0_i$ and $\Pi^0_i$ for arbitrary
$i$.
\end{abstract}

\begin{keyword}
Turing machine \sep recursively enumerable \sep decidable \sep approximation
\sep matching pennies \sep halting \sep halting problem
\sep elusive model paradox \sep red herring sequence \sep learnability
\sep Nash equilibrium \sep approximability \sep adversarial learning
\end{keyword}

\end{frontmatter}

\section{Introduction}
This paper deals with three widely-discussed topics: approximability,
conventional learnability and adversarial learnability,
and introduces a unified framework in which all three can be
studied.

First, consider approximability.
Turing's seminal 1936 result \citep{Turing:Computable} demonstrated that some
languages that can be accepted by Turing machines (TMs) are not decidable.
Otherwise stated, some R.E.\ languages are not recursive. Equivalently:
some co-R.E.\ languages are not R.E.; any R.E.\ language must
differ from them by at least one word. However, the diagonalisation process
by which this result was originally derived makes no stronger claim regarding
the number of words differentiating a co-R.E.\ language and an R.E.\ one.
It merely shows one example of a word where a difference must exist.

We extend this original result by showing that some co-R.E.\ languages are,
in some sense,
as different from any R.E.\ language as it is possible to be.

To formalise this statement, consider an arbitrary (computable) enumeration,
$w_1, w_2,\ldots$, over the complete language (the language that includes all
words over the chosen alphabet). Over this enumeration, $\{w_i\}$, we define
a distance metric,
\textbf{dissimilarity}, between two languages, $L_1$ and $L_2$, as follows.
\[
\text{DisSim}(L_1,L_2)\equiv\limsup_{n\to\infty}\frac{|(L_1 \triangle L_2)\cap\{w_1,\ldots,w_n\}|}{n},
\]
where $L_1 \triangle L_2$ is the symmetric difference.
We note that the value of $\text{DisSim}(L_1,L_2)$ depends on the enumeration
chosen, and therefore, technically, $\text{DisSim}(\cdot)=\text{DisSim}_{\{w_i\}}(\cdot)$.
However, all results in this paper are true for all possible
choices of the enumeration, for which reason we omit the choice of enumeration,
opting for this more simplified notation.

$\text{DisSim}(L_1,L_2)$ ranges between $0$ (the languages are essentially
identical) and $1$ (the languages are completely dissimilar).

We prove:
\begin{thm}\label{T:approx}
There is a co-R.E.\ language $\bar{L}$ such that every R.E.\ language
has a dissimilarity distance of $1$ from $\bar{L}$.
\end{thm}

Consider now learnability.
Learnability is an important concept in statistics, econometrics, machine
learning, inductive inference, data mining and other fields.
This has been discussed by E.\ M.\ Gold and by
L.\ G.\ Valiant in terms of language identification in the limit
\cite{Gold1967,Valiant1984}, and also in statistics via the notion of
statistical consistency, also known as ``completeness'' (converging arbitrarily
closely in the limit to an underlying true model).

Following upon his convergence results in \cite{Solomonoff1978},
Solomonoff writes \cite[sec. 2 (Completeness and Incomputability)]{Solomonoff2011}:
\begin{quote}
``It is notable that
completeness and incomputability are complementary properties: It is easy to
prove that any complete prediction method must be incomputable. Moreover,
any computable prediction method can not be complete -- there will always be a
large space of regularities for which its predictions are catastrophically poor.''
\end{quote}

In other words, in Solomonoff's problem set-up it is impossible for a Turing
machine to learn every R.E.\ language:
every computable learner is limited.

Nevertheless, in the somewhat different context within which we study
learnability, we are able to show that this tension does not exist:
a Turing machine can learn any computable language. Moreover, we will consider
a set of languages that includes, as a proper subset of it, the languages
$\Sigma^0_1 \cup \Pi^0_1$ and
will prove that while no deterministic learning algorithm can learn every
language in the set, a probabilistic one can (with probability $1$),
and a mixed strategy involving
several deterministic learning algorithms can approximate this arbitrarily
well.\footnote{Here and elsewhere we use the standard notations for language
families in the arithmetical hierarchy \cite{Rogers:recursive}: $\Sigma^0_1$ is
the set of recursively enumerable languages, $\Pi^0_1$ is the set of
co-R.E.\ languages.}

Lastly, consider adversarial learning \cite{Lowd:adversarial, Wei:adversarial,
Huang:adversarial}. This is different from the conventional
learning scenario described above in that while in conventional learning we
attempt to converge to an underlying ``true model'' based on given observations,
adversarial learning is a multi-player process in which each participant can
observe (to some extent) other players' predictions and adjust their own
actions accordingly. This game-theoretic set-up becomes of practical importance
in many scenarios. For example, in online bidding bidders use information
available to them (e.g., whether they won a particular auction) to learn the
strategy used by competing bidders, so as to be able to optimise their own
strategy accordingly.

We consider, specifically, an adversarial learning scenario in which one player
(the pursuer) attempts to copy a second player, while the second player
(the evader) is attempting to avoid being copied. Specifically, each player
generates a bit ($0$ or $1$) and the pursuer wins if the two bits are equal
while the
evader wins if they are not. Though on the face of it this scenario may seem
symmetric, we show that the pursuer has a winning strategy.

To attain all these results (as well as their higher-Turing-degree equivalents),
we introduce a unified framework in which these
questions and related ones can all be studied.
The set-up used is an adaptation of one initially introduced by Scriven
\cite{Scriven1965} of a predictor and a contrapredictive (or avoider)
effectively playing
what we might nowadays describe as a game of iterated matching pennies.
In Section~\ref{S:matching}, we give a formal description of this problem
set-up and briefly describe its historical evolution.
In Section~\ref{S:halting}, we explain the relevance of the set-up to
the learnability and approximability problems and analyse, as an
example case, adversarial learning in the class of
decidable languages. In Section~\ref{S:nonhalting}, we extend the
analysis to adversarial learning in all other classes in the arithmetical
hierarchy, and in particular to Turing machines.

In Sections~\ref{S:learnability} and \ref{S:approximability} we then return to
conventional learnability and to approximability, respectively, and prove the
remaining results by use of the set-up developed, showing how it can be
adapted to these problems.

\section{Matching Pennies}\label{S:matching}

The matching pennies game is a zero-sum two-player game where each player is
required to output a bit. If the two bits are equal, this is a win for
Player~``$=$''; if they differ, this is a win for
Player~``$\ne$''.
The game is a classic example used in teaching mixed strategies
\citep[see, e.g.][pp.\ 283--284]{flake1998computational}: its
only Nash equilibrium \citep{vonNeumann:1944:TGE,Nash1951:Non-cooperative_Games}
is a mixed strategy wherein each player
chooses each of the two options with probability $1/2$.

Consider, now, an iterative version of this game, where at each round the
players choose a new bit with perfect information of all previous rounds.
Here, too, the best strategy is to choose at each round a new bit with
probability $1/2$ for each option, and with the added caveat that each bit
must be independent of all previous bits.
In the iterative variation, we define the payoff (of the entire game) to be
\begin{equation}\label{Eq:ne}
S=S_{\ne}=\left(\liminf_{N\to\infty} \sum_{n=1}^{N} \frac{\delta_n}{2N}\right)+\left(\limsup_{N\to\infty} \sum_{n=1}^{N} \frac{\delta_n}{2N}\right)
\end{equation}
for Player~``$\ne$'', where $\delta_n$ is $0$ if the bits output in the $n$'th
round
are equal and $1$ if they are different. The payoff for Player~``$=$'' is
\begin{equation}\label{Eq:eq}
S_{=}=1-S_{\ne}=\left(\liminf_{N\to\infty} \sum_{n=1}^{N} \frac{1-\delta_n}{2N}\right)+\left(\limsup_{N\to\infty} \sum_{n=1}^{N} \frac{1-\delta_n}{2N}\right)
\end{equation}

These payoff functions were designed to satisfy
the following criteria:
\begin{itemize}
\item They are always defined.
\item The game is zero-sum and strategically symmetric, except for the
essential distinction between a player aiming to copy
(Player~``$=$'', the pursuer)
and a player aiming for dissimilarity (Player~``$\ne$'', the evader).
\item The payoff is a function solely of the $\{\delta_i\}$ sequence. (This is
important because in the actual IMP game being constructed players will only
have visibility into past $\delta_i$, not full information regarding the game's
evolution.)
\item Where a limit exists (in the $\lim$ sense) to the percentage of rounds to
be won by a player, the payoff is this percentage.
\end{itemize}
In particular, note that when the payoff functions take the value $0$ or $1$,
there exists a limit (in the $\lim$ sense) to the percentage of rounds to be
won by a player, and in this case the payoff is this limit.

In the case of the strategy pair described above, for example, where
bits are determined by independent, uniform-distribution coin tosses, the limit
exists and the payoff is $1/2$ for both players, indicating that the game is
not biased towards either. This is a Nash equilibrium of the game: neither
player can ensure a higher payoff for herself as long as the other persists
in the equilibrium strategy. The game has other Nash equilibria, but all share
the $(1/2,1/2)$ payoffs.

Above, we describe the players in the game as agents capable of randomisation:
they choose a random bit at each new round. However, the game can
be played, with the same strategies, also by deterministic agents. For this,
consider every possible infinite bit-string as a possible strategy for
each of the players. In this case, the game's Nash equilibrium would be a
strategy pair where each player allots a bit-string from a
uniform distribution among all options.

We formalise this deterministic outlook on the matching pennies game as follows.

\begin{defi}[Iterative Matching Pennies game]\label{D:IMP}
An \emph{Iterative Matching Pennies game} (or \emph{IMP}), denoted
$\text{IMP}(\Sigma_{=},\Sigma_{\ne})$,
is a two player game where
each player chooses a language: Player~``$=$'' chooses $L_{=}\in\Sigma_{=}$ and
Player~``$\ne$'' chooses $L_{\ne}\in\Sigma_{\ne}$, where $\Sigma_{=}$ and
$\Sigma_{\ne}$ are two collections of languages over the binary alphabet.

Where $\Sigma_{=}=\Sigma_{\ne} \;(=\Sigma)$, we denote the game
$\text{IMP}(\Sigma)$.

Define $\Delta_0$ to be the empty string and define, for every natural $i$,
\[
\delta_i \defeq \begin{cases}
1 & \text{if } \Delta_{i-1} \in L_{=}\triangle L_{\ne} \\
0 & \text{if } \Delta_{i-1} \not\in L_{=}\triangle L_{\ne}
\end{cases},
\]
\[
\Delta_i \defeq \Delta_{i-1}\delta_i,
\]
Then the payoffs $S_{=}=S_{=}(L_{=},L_{\ne})$ and
$S_{\ne}=S_{\ne}(L_{=},L_{\ne})$ are as
defined in \eqref{Eq:eq} and \eqref{Eq:ne}, respectively.
The notation ``$\Delta_{i-1}\delta_i$'' indicates string concatenation.

Player (mixed) strategies in this game are described as distributions,
$D_{=}$ and $D_{\ne}$, over $\Sigma_{=}$ and $\Sigma_{\ne}$, respectively.
In this case, we define
\[
S_{=}(D_{=},D_{\ne}) = E(S_{=}(L_{=},L_{\ne}))\quad L_{=} \sim D_{=}, L_{\ne} \sim D_{\ne}.
\]
\[
S_{\ne}(D_{=},D_{\ne}) = E(S_{\ne}(L_{=},L_{\ne}))\quad L_{=} \sim D_{=}, L_{\ne} \sim D_{\ne}.
\]
\end{defi}

Note again that the game is zero sum:
any pair of strategies, pure or mixed, satisfies
\begin{equation}\label{Eq:scoresum}
S_{=}(D_{=},D_{\ne})+S_{\ne}(D_{=},D_{\ne})= 1.
\end{equation}

To better illustrate the dynamics embodied by Definition~\ref{D:IMP}, let us
add two more definitions: let
\begin{equation}\label{Eq:O_eq}
O_=(i) \defeq \begin{cases}
1 & \text{if } \Delta_{i-1} \in L_{=} \\
0 & \text{if } \Delta_{i-1} \not\in L_{=}
\end{cases}
\end{equation}
and let
\begin{equation}\label{Eq:O_ne}
O_{\ne}(i) \defeq \begin{cases}
1 & \text{if } \Delta_{i-1} \in L_{\ne} \\
0 & \text{if } \Delta_{i-1} \not\in L_{\ne}
\end{cases},
\end{equation}
noting that by Definition~\ref{D:IMP},
$\delta_i=O_=(i) \oplus O_{\ne}(i)$, where ``$\oplus$'' denotes the exclusive or
(``xor'') function.

The scenario encapsulated by the IMP game is that of a competition between two
players, Player~``$=$'' and Player~``$\ne$'', where the strategy of the players
is encoded in the form of the languages $L_{=}$ and $L_{\ne}$, respectively
(or distributions over these in the case of mixed strategies).

After $i$ rounds, each player has visibility to the set of results so far.
This is encoded by means of $\Delta_i$, a word composed of the characters
$\delta_1,\ldots, \delta_i$, where each $\delta_k$ is $0$ if the bits that
were output by the two players in round $k$ are equal and $1$ if they are not.
It is based on this history that the players now generate a new bit:
Player~``$=$'' generates $O_=(i+1)$ and Player~``$\ne$'' generates
$O_{\ne}(i+1)$.
The players' strategies are therefore functions from a word ($\Delta_i$) to
a bit ($O_{=}(i+1)$ for Player~``$=$'', $O_{\ne}(i+1)$ for Player~``$\ne$'').
To encode these strategies in the most general form, we use
languages: $L_=$ and $L_{\ne}$ are simply sets containing all the words to
which the response is ``$1$''.
Our choice of how weak or how strong a player can be is then ultimately in
the question of what language family, $\Sigma$, its strategy is chosen from.

Once $O_=(i+1)$ and $O_{\ne}(i+1)$ are determined, $\delta_{i+1}$ is simply
their xor ($1$ if the bits differ, $0$ if they are the same), and in this way
the definition generates the infinite list of $\delta_i$ that is ultimately
used to compute the game's overall payoff for each player.

Were we to actually try and run a real-world IMP
competition by directly implementing the definitions above, and were we to try
to implement the Nash equilibrium player strategies, we would immediately run
into two elements in the set-up that are incomputable: first, the choice of
a uniform infinitely-long bit-string, our chosen distribution among the
potential strategies, is incomputable (it is a choice among
uncountably many
elements); second, for a deterministic player (an agent) to output all the bits
of an arbitrary (i.e., general) bit-string, that player cannot be a Turing
machine. There are only countably many Turing machines, so only countably many
bit-strings that can thus be output.

In this paper, we examine the IMP game with several choices for
$\Sigma_{=}$ and $\Sigma_{\ne}$. The main case studied is where
$\Sigma_{=}=\Sigma_{\ne}=\Sigma^0_1$. In this
case, we still allow player mixed strategies to be incomputable distributions,
but any $L_{=}$ and $L_{\ne}$ are computable by TMs.

The set-up described here, where Iterated Matching Pennies is essentially
described as
a pursuit-evasion game, was initially introduced informally by Scriven
\cite{Scriven1965} in order to prove that unpredictability is innate to humans.
Lewis and Richardson \cite{LewisRichardson1966}, without explicitly mentioning
Turing machines or any (equivalent) models of computation, reinvestigated the
model and used it to refute Scriven's claim, with a proof that hinges on the
halting problem, but references it only implicitly.

The set-up was redeveloped independently by Dowe, first in the context of the
avoider trying to choose the next number in an integer sequence to be larger
(by one) than the (otherwise) best inference that one might expect
\cite[sec. 0.2.7, p.\ 545, col.\ 2 and footnote 211]{Dowe2008a},
and then, as in \cite{Scriven1965}, in the context of predicting bits in a
sequence \cite[p.\ 455]{Dowe2008b}\cite[pp.\ 16--17]{Dowe2013a}.
Dowe was the first to introduce the terminology of TMs into the set-up. His
aim was to illicit a paradox, which he dubbed ``the elusive model paradox'',
whose resolution relies on the undecidability of the halting problem. Thus, it
would provide an alternative to the method of \cite{Turing:Computable} to prove
this undecidability.
Variants of the elusive model paradox and of the ``red herring sequence''
(the optimal sequence to be used by an avoider)
are discussed in \cite[sec.\ 7.5]{Dowe2011a}, with the paradox also mentioned in
\cite[sec.\ 2.2]{DoweHernandez-OralloDas2011}\cite[footnote 9]{Hernandez-OralloDoweEspana-CubilloHernandez-LloredaInsa-Cabrera2011}.

Yet a third independent incarnation of the model was by Solomonoff, who
discussed variants of the elusive model paradox and the red herring sequence
in \cite[Appendix B]{Solomonoff2010} and \cite[sec.\ 3]{Solomonoff2009}.

We note that the more formal investigations of Dowe and of Solomonoff were in
contexts in which the ``game'' character of the set-up was not explored.
Rather, the set-up was effectively a one-player game, where regardless of the
player's choice of next bit, the red herring sequence's next bit was its
reverse. We, on the other hand, return to the original spirit of Scriven's
formulation, investigating the dynamics of the two player game, but do so
in a formal setting.

Specifically, we investigate the question of which of the two players
(if either) has an advantage in this
game, and, in particular, we will be interested in the game's Nash equilibria,
which are the pairs of strategies $(D_{=}^*, D_{\ne}^*)$ for which
\[
S_{=}(D_{=}^*,D_{\ne}^*) = \sup_{D_{=}} S_{=}(D_{=},D_{\ne}^*)
\]
and
\[
S_{\ne}(D_{=}^*,D_{\ne}^*) = \sup_{D_{\ne}} S_{\ne}(D_{=}^*,D_{\ne}).
\]
We define
\[
\text{minmax}(\Sigma_{=},\Sigma_{\ne}) = \inf_{D_{=}} \sup_{D_{\ne}} S(D_{=},D_{\ne})
\]
and
\[
\text{maxmin}(\Sigma_{=},\Sigma_{\ne}) = \sup_{D_{\ne}} \inf_{D_{=}} S(D_{=},D_{\ne}),
\]
where $D_{=}$ is a (potentially incomputable) distribution over $\Sigma_{=}$
and $D_{\ne}$ is a (potentially incomputable) distribution over $\Sigma_{\ne}$.
Where $\Sigma_{=}=\Sigma_{\ne} \;(=\Sigma)$, we will abbreviate this to
$\text{minmax}(\Sigma)$ and $\text{maxmin}(\Sigma)$.

A Nash equilibrium $(D_{=}^*, D_{\ne}^*)$ must satisfy
\begin{equation}\label{Eq:minmax}
S(D_{=}^*,D_{\ne}^*) = \text{maxmin}(\Sigma_{=},\Sigma_{\ne}) = \text{minmax}(\Sigma_{=},\Sigma_{\ne}),
\end{equation}
where, as before, $S=S_{\ne}$.

We note that while it may seem, at first glance, that the introduction of game
dynamics
into the problems of learnability and approximability inserts an unnecessary
complication into their analysis, in fact, we will show that the ability to
learn and/or approximate languages, when worded formally, involves a large
number of interlocking ``$\lim$'', ``$\sup$'', ``$\inf$'', ``$\limsup$'' and
``$\liminf$'' clauses that are most naturally expressed in terms of minmax
and maxmin solutions, Nash equilibria and mixed strategies.

\section{Halting Turing machines}\label{S:halting}

The IMP game serves as a natural platform for investigating
adversarial learning:
each of the players has the opportunity to learn from
all previous rounds, extrapolate from this to the question of what algorithm
their adversary is employing and then choose their own course of action to
best counteract the adversary's methods.

Furthermore, where $\Sigma_{=}=\Sigma_{\ne} \;(=\Sigma)$,
IMP serves as a natural
arena to differentiate between the learning
of a language (e.g., one selected from R.E.) and its complement (e.g., a
language selected from co-R.E.), because Player~``$=$'', the copying player, is
essentially trying to learn a language from $\Sigma$, namely that chosen by
Player~``$\ne$'', whereas Player~``$\ne$'' is attempting to learn a language from
co-$\Sigma$, namely the complement to that chosen by Player~``$=$''.
Any advantage
to
Player~``$=$'' can be attributed solely to the difficulty to learn
co-$\Sigma$ by an algorithm from $\Sigma$, as opposed to the ability to learn
$\Sigma$.

To exemplify IMP analysis, consider first the game where $\Sigma=\Delta^0_1$,
the set of decidable languages.
Because decidable languages are a set known
to be closed under complement, we expect Player~``$\ne$'' to be equally as
successful as
Player~``$=$'' in this variation. Consider, therefore, what would be the Nash
equilibria in this case.

\begin{thm}\label{T:halting}
Let $\Sigma$ be the set of decidable languages over $\{0,1\}^{*}$.
The game $\text{IMP}(\Sigma)$ does not have any Nash equilibria.
\end{thm}

We remark here that most familiar and typically-studied games belong to a
family of games where the space of mixed strategies is compact and convex,
such as those having a finite number of pure strategies, and
such games necessarily have at least one
Nash equilibrium. However, the same is not true for arbitrary games. (For
example, the game of ``guess the highest number'' does not have a Nash
equilibrium.) IMP, specifically, does not belong to a game family that
guarantees the existence of Nash equilibria.

\begin{proof}
We begin by showing that for any (mixed) strategy $D_{\ne}$,
\begin{equation}\label{Eq:sup}
\sup_{D_{=}} S_{=}(D_{=},D_{\ne})=1.
\end{equation}

Let $\mathcal{T}_0,\mathcal{T}_1, \ldots$ be any (necessarily incomputable)
enumeration over those Turing machines that halt on every input, and
let $L_0, L_1, \ldots$ be the sequence of languages that is accepted by them.
The sequence $\{L_i\}$ enumerates (with repetitions) over all languages in
$\Sigma=\Delta^0_1$.
Under this enumeration we have
\[
\lim_{X\to\infty}\text{Prob}(\exists x\le X\text{, such that }L_{\ne}=L_x)=1;\quad L_{\ne}\sim D_{\ne}.
\]
For this reason, for any $\epsilon$ there exists an $X$
such that
\[
\text{Prob}(\exists x \le X\text{, such that }L_{\ne}=L_x)\ge 1-\epsilon;\quad L_{\ne}\sim D_{\ne}.
\]

We devise a strategy, $D_{=}$, to be used by Player~{=}. This strategy will be
pure: the player will always choose language $L_{=}$, which we will now
describe. The language $L_{=}$ is the one accepted by Algorithm~\ref{A:mixed}.

\begin{algorithm}
\caption{Algorithm for learning a mixed strategy}
\label{A:mixed}
\begin{algorithmic}[1]
\Function {calculate bit}{$\Delta$}
\State $d \gets \|\Delta\|_1$.
\Comment Number of prediction errors so far.
\If {$d>X$}
\State Accept.
\ElsIf {$\Delta \in L_d$}\label{Step:if_in}
\State Accept.
\Else
\State Reject.
\EndIf
\EndFunction
\end{algorithmic}
\end{algorithm}

Note that while the enumeration $\mathcal{T}_0,\mathcal{T}_1,\ldots$ is not
computable, Algorithm~\ref{A:mixed} only requires
$\mathcal{T}_0,\ldots,\mathcal{T}_X$ to be accessible to it,
and this can be done because any such finite set of TMs can be hard coded
into Algorithm~\ref{A:mixed}.

Consider the game,
on the assumption that Player~``$\ne$'''s strategy is $L_x$ for
$x\le X$.
After at most $x$ prediction errors, Algorithm~\ref{A:mixed} will begin
mimicking a strategy
equivalent to $L_x$ and will win every round from that point on.

We see, therefore, that for any $x\in\{0,\ldots,X\}$ we have
$S_{=}(L_{=},L_x)=1$, from which we conclude
that $S_{=}(D_{=},D_{\ne})\ge 1-\epsilon$ (or, equivalently,
$S_{\ne}(D_{=},D_{\ne})\le \epsilon$), in turn proving
that for any Nash equilibrium $(D_{=}^*,D_{\ne}^*)$ we necessarily must have
\begin{equation}\label{Eq:maxmin0}
\text{maxmin}(\Sigma)=0.
\end{equation}

For exactly the symmetric reasons, when $\Sigma=\Delta^0_1$ we also have
\begin{equation}\label{Eq:minmax1}
\text{minmax}(\Sigma)=1:
\end{equation}
Player~``$\ne$'' can follow a strategy identical to that described in
Algorithm~\ref{A:mixed}, except reversing the condition in
Step~\ref{Step:if_in}.

Because we now have that $\text{minmax}(\Sigma)\ne\text{maxmin}(\Sigma)$,
we know that Equation~\eqref{Eq:minmax} cannot be satisfied for any
strategy pair. In particular, there are no Nash equilibria.
\end{proof}

This result is not restricted to $\Sigma=\Delta^0_1$, the decidable languages,
but also to any set of languages that is powerful enough to encode
Algorithm~\ref{A:mixed} and its complement. It is true, for example, for
$\Delta^0_0$ as well as for $\Delta^0_1$ with any set of Oracles, i.e.,
specifically, for any $\Delta^0_i$.

\begin{defi}
We say that a collection of languages $\Sigma_{\ne}$ is
\textbf{adversarially learnable}
by a collection of strategies $\Sigma_{=}$ if
$\text{minmax}(\Sigma_{=},\Sigma_{\ne})=0$.

If a collection is adversarially learnable by $\Sigma^0_1$, we simply say that
it is \textbf{adversarially learnable}.
\end{defi}

\begin{cor}
$\forall i, \Delta^0_i$ is not adversarially learnable by $\Delta^0_i$.
\end{cor}

\begin{proof}
As was shown in the proof of Theorem~\ref{T:halting},
$\text{minmax}(\Delta^0_i,\Delta^0_i)=1$.
\end{proof}

We proceed, therefore, to the question of how well each player fares when
$\Sigma$ includes non-decidable R.E.\ languages, and is therefore no longer
closed under complement.

\section{Adversarial learning}\label{S:nonhalting}

We claim that R.E.\ languages are adversarially learnable, and that it is
therefore not possible to learn the complement of R.E.\ languages
in general, in the adversarial learning scenario.

\begin{thm}\label{T:nonhalting}
The game $\text{IMP}(\Sigma^0_1)$ has a strategy, $L_{=}$, for Player~``$=$''
that guarantees $S_{\ne}(L_{=},L_{\ne})=0$ for all $L_{\ne}$
(and, consequently, also for all distributions
among potential $L_{\ne}$ candidates).

In particular, $\Sigma^0_1$ is adversarially learnable.
\end{thm}

\begin{proof}
We describe $L_{=}$ explicitly by means of an algorithm accepting it. This is
given in Algorithm~\ref{A:pursuer}.

\begin{algorithm}
\caption{Algorithm for learning an R.E.\ language}
\label{A:pursuer}
\begin{algorithmic}[1]
\Function{calculate bit}{$\Delta$}
\State Let $\mathcal{T}_0, \mathcal{T}_1, \ldots$ be an enumeration over all
Turing machines.
\State $d \gets \|\Delta\|_1$.
\Comment Number of prediction errors so far.
\State Simulate $\mathcal{T}_d$
\EndFunction
\end{algorithmic}
\end{algorithm}

Note that Algorithm~\ref{A:pursuer} does not have any ``Accept'' or ``Reject''
statements. It returns a bit only if $\mathcal{T}_d$ returns a bit and does not
terminate if $\mathcal{T}_d$ fails to terminate. To actually simulate
$\mathcal{T}_d$ and to encode the enumeration $\mathcal{T}_0, \ldots$,
Algorithm~\ref{A:pursuer} can simply use a universal Turing machine,
$\mathcal{U}$, and define the enumeration in a way such that
$\mathcal{U}$ accepts the input ``$d\#\Delta$'' if and only if $\mathcal{T}_d$
accepts the input $\Delta$.

To show that Algorithm~\ref{A:pursuer} cannot be countered, consider any
R.E.\ language to be chosen by Player~``$\ne$''. This language, $L_{\ne}$,
necessarily corresponds to the output of $\mathcal{T}_x$ for some (finite) $x$.
In total, Player~``$=$'' can lose at most $x$ rounds.
In every subsequent round, its output will be identical to that of
$\mathcal{T}_x$, and therefore identical to the bit chosen by Player~``$\ne$''.
\end{proof}

We see, therefore, that the complement of Algorithm~\ref{A:pursuer}'s
language cannot be learned by any
R.E.\ language. Player~``$\ne$'' cannot hope to win more than a finite number of
rounds.

Note that these results do not necessitate that $\Sigma=\Sigma^0_1$, the
R.E.\ languages. As long as $\Sigma$ is rich enough to allow implementing
Algorithm~\ref{A:pursuer}, the results hold. This is true, for example, for
$\Sigma$ sets that allow Oracle calls. In particular:

\begin{cor}\label{C:i_adversarial}
For all $i>0$, $\Sigma^0_i$ is adversarially learnable by $\Sigma^0_i$ but not
by $\Pi^0_i$;
$\Pi^0_i$ is adversarially learnable by $\Pi^0_i$ but not by $\Sigma^0_i$.
\end{cor}

\begin{proof}
To show the learnability results, we use Algorithm~\ref{A:pursuer}. To show
the non-learnability results, we appeal to the symmetric nature of the game:
if Player~``$=$'' has a winning learning strategy, Player~``$\ne$'' does not.
\end{proof}

\section{Conventional learnability}\label{S:learnability}

To adapt the IMP game for the study of conventional
(i.e., non-adversarial) learning and approximation,
we introduce the notion of \textbf{nonadaptive strategies}.

\begin{defi}\label{D:NA}
A \textbf{nonadaptive strategy} is a language, $L$, over $\{0,1\}^*$ such that
\[
\forall u, v, |u|=|v| \Rightarrow (u\in L \Leftrightarrow v\in L),
\]
where $|u|$ is the bit length of $u$.

Respective to an arbitrarily chosen (computable) enumeration $w_1, w_2,\ldots$
over the complete language, we define the function $\textit{NA}()$ such that,
for any language $L$, $\textit{NA}(L)$ is the language such that
\[
x\in \textit{NA}(L) \Leftrightarrow w_{|x|} \in L.
\]

Furthermore, for any collection of languages, $\Sigma$, we define
$\textit{NA}(\Sigma) = \{\textit{NA}(L) | L \in \Sigma\}.$

$\textit{NA}(\Sigma)$ is the \textbf{nonadaptive application} of $\Sigma$.
\end{defi}

To elucidate this definition,
consider once again a (computable) enumeration, $w_1, w_2,\ldots$ over the
complete language.

In previous sections, we have analysed the case where the two competing
strategies are adaptive (i.e., general). This was the case of adversarial
learning. Modelling the conventional learning problem is simply done by
restricting
$\Sigma_{\ne}$ to nonadaptive strategies. The question of whether a strategy
$L_{=}$ (or $D_{=}$) can learn $L$ is the question of whether it can learn
adversarially $\textit{NA}(L)$.
The reason this is so is because the bit output at any
round $i$ by a nonadaptive strategy is independent of any response made by
either player at any previous round: at each round $i$, $O_{\ne}(i+1)$, the
response of Player~``$\ne$'', as defined in \eqref{Eq:O_ne}, is a function of
$\Delta_i$, a word composed of exactly $i$ bits. Definition~\ref{D:NA} now adds
to this the restriction that the response must be invariant to the value of
these $i$ bits and must depend only on the bit length, $i$, which is to say
on the round number. Regardless of what the strategy of Player~``$=$'' is,
the sequence $O_{\ne}(1), O_{\ne}(2),\ldots$ output by Player~``$\ne$'' will
always remain the same. Thus, a nonadaptive strategy for Player~``$\ne$'' is
one where the player's output is a predetermined, fixed string of bits, and
it is this string that the opposing strategy of Player~``$=$''
must learn to mimic.

Note, furthermore, that if $\Sigma_{\textit{NA}}$ is the set of all nonadaptive
languages, then for every $i>0$ we have
\begin{equation}\label{Eq:na_sigma}
\textit{NA}(\Sigma^0_i) = \Sigma^0_i \cap \Sigma_{\textit{NA}}.
\end{equation}
The equality stems from the fact that calculating $w_{|x|}$ from $x$ and
vice versa (finding any $x$ that matches $w_{|x|}$) is, by definition,
recursive, so there
is a reduction from any $L$ to $\textit{NA}(L)$ and back.
If a language can
be computed over the input $w_{|x|}$ by means of a certain nonempty set of
quantifiers, no additional unbounded quantifiers are needed to compute it
from $x$.

This leads us to Definition~\ref{D:learn}.

\begin{defi}\label{D:learn}
We say that a collection of languages $\Sigma_{\ne}$ is (conventionally)
\textbf{learnable} by a collection of strategies $\Sigma_{=}$ if
$\text{minmax}\left(\Sigma_{=},\textit{NA}(\Sigma_{\ne})\right)=0$.

If a collection is learnable by $\Sigma^0_1$, we simply say that it is
\textbf{learnable}.
\end{defi}

\begin{cor}\label{C:i_learnable}
For all $i>0$, $\Sigma^0_i$ is learnable by $\Sigma^0_i$. In particular,
$\Sigma^0_1$ is learnable.
\end{cor}

\begin{proof}
We have already shown (Corollary~\ref{C:i_adversarial})
that $\Sigma^0_i$ is adversarially learnable by $\Sigma^0_i$, and
$\textit{NA}(\Sigma^0_i)$ is a subset of $\Sigma^0_i$, as demonstrated by
\eqref{Eq:na_sigma}.
\end{proof}

Constraining Player~``$\ne$'' to only be able to choose nonadaptive strategies
can only lower the $\text{minmax}$ value. Because it is already at $0$, it
makes no change: we are weakening the player that is already weaker.
It is more rewarding to constrain Player~``$=$'' and to consider the game
$\text{IMP}\left(\textit{NA}(\Sigma^0_i),\Sigma^0_i\right)$. Note, however, that
this is equivalent to the game
$\text{IMP}\left(\Sigma^0_i,\textit{NA}(\Pi^0_i)\right)$ under role reversal.

\begin{thm}\label{T:pi_learn}
$\Pi^0_1$ is learnable.
\end{thm}

\begin{proof}

To begin, let us consider a simpler scenario than was discussed so far.
Specifically, we will consider a scenario in which the feedback available to
the learning algorithm at each point is not only $\Delta_n$, the information of
which rounds it had ``won'' and which it had ``lost'', but also
$O_{=}(n)$ and $O_{\ne}(n)$, what the bit output by each machine was, at
every step.\footnote{Because $O_{=}(n)\oplus O_{\ne}(n)\oplus \delta_n=0$,
using any two of these as input to the TM is equivalent to using all three,
because the third can always be calculated from the others.}

In this scenario, Player~``$=$'' can calculate a co-R.E.\ function by
calculating its complement in round $n$ and then reading the result as the
complement to $O_{=}(n)$, which is given to it in all later rounds.

For example, at round $n$ Player~``$=$'' may simulate a particular Turing
machine, $\mathcal{T}$, in order to test whether it halts. If it does halt,
the player halts and accepts the input, but it may also continue indefinitely.
The end effect is that if $\mathcal{T}$ halts then $O_{=}(n)=1$ and otherwise
it is
$0$. At round $n+1$, Player~``$=$'' gets new inputs. (Recall that if one views
the player as a Turing machine, it is effectively restarted at each round.)
The new input in the real IMP game is $\Delta_n$, but for the moment we
are assuming a simpler version where the input is the pair of strings
$(O_{=}(1)\ldots O_{=}(n),O_{\ne}(1)\ldots O_{\ne}(n))$. This being the case,
though whether $\mathcal{T}$ halts or not is in general not computable by a
$\Sigma^0_1$ player, once a simulation of the type described here is run at
round $n$, starting with round $n+1$ the answer is available to the player in
the form of $O_{=}(n)$, which forms part of its input.

More concretely, one algorithm employable by Player~``$=$'' against a known
nonadaptive language $\textit{NA}(L_{\ne})$ is one that calculates
``$w_{2n+1} \notin L_{\ne}$?'' (which is an R.E.\ function)
in every $2n$'th round, and then uses this
information in the next round in order to make the correct prediction.
This guarantees $S\left(L_{=},\textit{NA}(L_{\ne})\right)\le 1/2$.
However, it is possible to do better.

To demonstrate how, consider that Player~``$=$'' can determine the answer to the
question ``$|\{w_i,\ldots,w_j\}\setminus L_{\ne}| \ge k$?'' for any chosen
$i$, $j$ and $k$. The way to do this is to simulate simultaneously all $j+1-i$
Turing machine runs that calculate ``$w_l \notin L_{\ne}$?'' for each
$i\le l \le j$ and to halt if $k$ of them halt.
As with the previous example, by performing this algorithm at any stage $n$,
the algorithm will then be able to read out the result as $O_{=}(n)$
in all later rounds.

Consider, now, that this ability can be used to determine
$|\{w_i,\ldots,w_j\}\setminus L_{\ne}|$ exactly (rather than simply bounding it)
by means of a binary search, starting with the question
``$|\{w_i,\ldots,w_j\}\setminus L_{\ne}| \ge 2^{m-1}$?''
in the first round, and proceeding to increasingly finer determination of the
actual set size on each later round. Player~``$=$'' can therefore determine the
number of ``$1$'' bits in a set of $j+1-i=2^m-1$ outputs of a co-R.E.\ function
in this way in only $m$ queries,
after which the number will be written in binary form, from most significant
bit to least significant bit, in its $O_{=}$ input.
Once this cardinality has been
determined, Player~``$=$'' can compute via a terminating computation the value
of each of ``$w_l \in L_{\ne}$?'': the player will simulate, in
parallel, all $j+1-i$ machines, and will terminate the computation either
when the desired bit value is found via a halting of the corresponding
machine, or until the full cardinality of halting machines has been
reached, at which point, if the desired bit is not among the machines that
halted, then the player can safely conclude that its computation will never
halt.

Let $\{m_t\}$ be an arbitrary (computable) sequence with
$\lim_{t \to \infty} m_t = \infty$.
If Player~``$=$'' repeatedly
uses $m_t$ bits (each time picking the next value in the sequence) of its own
output in order to determine Player~``$\ne$'''s next $2^{m_t}-1$ bits, the
proportion of bits determined correctly by this will approach $1$.

However, the actual problem at hand is one where Player~``$=$'' does not have
access to its own output bits, $\left(O_{=}(1),\ldots,O_{=}(n)\right)$. Rather,
it can only see $\left(\delta_1,\ldots,\delta_n\right)$, the exclusive or (xor)
values of its bits and those of Player~``$\ne$''. To deal with this situation,
we use a variation over the strategy described above.

First, for convenience, assume that Player~``$=$'' knows the first $m_0$ bits
to be output by Player~``$\ne$''.
Knowing Player~``$\ne$'''s bits and having visibility as to
whether they are the same or different to Player~``$=$'''s bits give, together,
Player~``$=$'' access to its own past bits.

Now, it can use these first $m_0$ bits in order to encode, as before, the
cardinality of the next $2^{m_0}-1$ bits, and by this also their individual
values (as was demonstrated previously with the calculation of
``$w_l\in L_{\ne}$?'').
This now gives Player~``$=$'' the ability to win every one of the
next $2^{m_0}-1$ rounds. However, instead of utilising this ability to the
limit, Player~``$=$'' will only choose to win the next
$2^{m_0}-1-m_1$, leaving the remaining $m_1$ bits free to be used for
encoding the cardinality of the next $2^{m_1}-1$. This strategy can be
continued to all $m_t$.
The full list of criteria required of the sequence $\{m_t\}$ for this
construction to work and to ultimately lead to
$S\left(L_{=},\textit{NA}(L_{\ne})\right)=0$ is:
\begin{enumerate}
\item $\lim_{t \to \infty} m_t=\infty$.
\item $\forall t, m_{t+1}\le 2^{m_t}-1$.
\item $\lim_{t \to \infty} \frac{m_{t+1}}{2^{m_t}}=0$.
\end{enumerate}
A sequence satisfying all these criteria can easily be found,
e.g.\ $m_t=t+2$.

Two problems remain to be solved: (1) How to determine the value of the first
$m_0$ bits, and (2) how to deal with the fact that $L_{\ne}$ is not known.

We begin by tackling the second of these problems. Because $L_{\ne}$ is not
known, we utilise a strategy of enumerating over the possible languages,
similar to what is done in Algorithm~\ref{A:pursuer}. That is to say, we
begin by assuming that co-$L_{\ne}=L_0$ and respond accordingly. Then, if we
detect that the responses from Player~``$\ne$'' do not match those of $L_0$ we
progress to assume that co-$L_{\ne}=L_1$, etc.. We are not always in a position
to tell if our current hypothesis of $L_{\ne}$ is correct, but we can verify
that it matches at least the first $2^{m_t}-m_{t+1}-1$ bits of each $2^{m_t}-1$
set. If Player~``$=$'' makes any incorrect
predictions during any of these $2^{m_t}-m_{t+1}-1$ rounds, it can progress to
the next hypothesis. We note that it is true that Player~``$=$'' can remain
mistaken about the identity of $L_{\ne}$ forever, as long as $L_{\ne}$ is such
that the first $2^{m_t}-m_{t+1}-1$ predictions of every $2^{m_t}-1$ are correct,
but because these correct predictions alone are enough to ensure
$S\left(L_{=},\textit{NA}(L_{\ne})\right)=0$, the
question of whether the correct $L_{\ne}$ is ultimately found or not is moot.

To tackle the remaining problem, that of determining $m_0$ bits of $L_{=}$ in
order to bootstrap the process, we make use of mixed strategies.

Consider a mixed strategy involving probability $1/2^{m_0}$ for each of
$2^{m_0}$ strategies, differing only by the $m_0$ bits they assign as the first
bits for each language in order to bootstrap the learning process.
If co-$L_{\ne}=L_0$, of the $2^{m_0}$ strategies one will make the correct guess
regarding the first $m_0$ input bits, after which that strategy can ensure
$S\left(L_{=},\textit{NA}(L_{\ne})\right)=0$.
However, note that, if implemented as described so far,
this is not the case for any other $L_i$. Suppose, for example, that
co-$L_{\ne}=L_1$. All $2^{m_0}$ strategies begin by assuming, falsely, that
co-$L_{\ne}=L_0$, and all may discover later on that this assumption is
incorrect, but they may do so at different rounds. Because of this, a
counter-strategy can be designed to fool all $2^{m_0}$ learner strategies.

To avoid this pitfall, all strategies must use the same bit positions in
order to bootstrap learning for each $L_i$, so these bit positions must be
pre-allocated. We will use bits $a^2,\ldots,a^2+m_0-1$ in order to bootstrap
the $i$'th hypothesis, for some known $a=a(i)$, regardless of whether the
hypothesis $L_{=}=L_i$ is known to require checking
before these rounds, after, or not at all.
The full set of rounds pre-allocated in this way still has only density zero
among the integers, so even without a win for
Player~``$=$'' in any of these rounds its final payoff remains $1$.

Suppose, now, that $L_i$ is still not the assumption currently being verified
(or falsified) at rounds $a^2,\ldots,a^2+m_0-1$. The Hamming weight (number of
``$1$''s)
 of which $2^{m_0}-1$
bits should be encoded by Player~``$=$'' in these rounds' bits? To solve this,
we will pre-allocate to each hypothesis an infinite number of bit positions,
which, altogether for all hypotheses, still amount to a set of density $0$ among
the integers. The hypothesis
will continuously predict the values of this pre-allocated infinite sequence
of bits until it becomes the ``active'' assumption. If and when it does, it
will expand its predictions to all remaining bit positions.

This combination of $2^{m_0}$ strategies, of which one guarantees a payoff of
$1$, therefore guarantees in total an expected payoff of at least $1/2^{m_0}$.
We want to show, however, that
$\text{minmax}\left(\Sigma^0_1,\textit{NA}(\Pi^0_1)\right)=0$.
To raise from $1/2^{m_0}$
to $1$, we describe a sequence of mixed strategies for which
the expected payoff for Player~``$=$'' converges to $1$.

The $k$'th element in the sequence of mixed strategies will be composed of
$2^{m_0 k}$ equal probability pure strategies. The strategies will follow the
algorithm so far, but instead of moving from the hypothesis co-$L_{\ne}=L_i$ to
co-$L_{\ne}=L_{i+1}$ after a single failed attempt (which may be due to
incorrect bootstrap bits), the algorithm will try each $L_i$ language $k$
times. In total, it will guess at most $m_0 k$ bits for each language, which are
the $m_0 k$ bits defining the strategy.

This strategy ensures a payoff of at least $1-(1-1/2^{m_0})^k$, so converges to
$1$, as desired, for an asymptotically large $k$.

The full algorithm is described in Algorithm~\ref{A:universal_learning}.
It uses the function $\textit{triangle}$, defined as follows: let
\[
\textit{base}(x)=\left\lfloor\frac{\lfloor\sqrt{8x+1}\rfloor-1}{2}\right\rfloor
\]
and
\begin{equation}\label{Eq:triangle}
\textit{triangle}(x)=x-\textit{base}(x)(\textit{base}(x)+1)/2.
\end{equation}
The value of $\textit{triangle}(x)$ for $x=0, 1, 2, \ldots$ equals
\[
0, 0, 1, 0, 1, 2, 0, 1, 2, 3, 0, 1, 2, 3, 4, 0, 1, 2, 3, 4, 5,\ldots,
\]
describing a triangular walk through the nonnegative integers.

The algorithm is divided into two stages. In Step 1, the algorithm simulates
its actions in all previous rounds, but without simulating any (potentially
non-halting) Turing machine associated with any hypothesis. The purpose of
this step is to determine which hypothesis (choice of Turing machine and
bootstrapping) is to be used for predicting the next bit. Once the hypothesis
is determined, Step 2 once again simulates all previous rounds, only this time
simulating the chosen hypothesis wherever it is the active hypothesis. In this
way, the next bit predicted by the hypothesis can be determined.

The specific $\{m_t\}$ sequence used in Algorithm~\ref{A:universal_learning}
is $m_t=t+2$ (which was previously mentioned as an example of a sequence
satisfying all necessary criteria).

\begin{algorithm}
\caption{Algorithm for learning any co-R.E.\ language}
\label{A:universal_learning}
\begin{algorithmic}[1]
\State \Comment The strategy is a uniform mixture of $4^k$ algorithms.
\State \Comment We describe the $j$'th algorithm.
\Function {calculate bit}{$\Delta$}
\State $n \gets$ length of $\Delta$
\Comment The round number. Let $\Delta=\delta_1,\ldots,\delta_n$.
\State \Comment Step 1: Identify $h$, the current hypothesis.
\State $\textit{NonActiveHypotheses} \gets \{\}$
\State $\textit{PredPos} \gets \{\}$
\Comment A set managing which positions are predicted by which hypothesis.
\For {$i \in 0,\ldots,n$}
\If {$\exists (h,S,S')\in\textit{PredPos}$ such that $i\in S$}
\State Let $h,S,S'$ be as above.
\State \Comment $h=$ hypothesis number.
\State \Comment $S=$ predicted positions.
\State \Comment $S'=$ next positions to be predicted.

\State Let $m$ be such that $2^{m-1}-1=|S'|$.
\Comment We only construct $S'$ that have such an $m$.

\ElsIf {$\exists a,h$ such that $a^2=i$, $h=\textit{triangle}(a)$ and $h\notin\textit{NonActiveHypotheses}$}
\State \Comment First bootstrap bit for hypothesis $h$.
\State Let $h$ be as above.
\State $S \gets \{\}$
\State $S' \gets \{i,i+1\}$
\State $\textit{bootstrap}(h) \gets i$
\State $m \gets 2$

\ElsIf {$i=n$}
\Comment Unusable bits.
\State Accept input.
\Comment Arbitrary choice.
\Else
\State Next $i$.
\EndIf

\State $e \gets |\{x\in S| x>i\}|$
\If {$e\ge m$}
\State \Comment These bits are predicted accurately for the correct hypothesis.
\If {$i< n$ and $\delta_{i+1}=1$}
\State \Comment Incorrect prediction, so hypothesis is false.
\State $\textit{NonActiveHypotheses} \gets \textit{NonActiveHypotheses}\cup\{h\}$
\State $\textit{PredPos} \gets \{(\tilde{h},\tilde{S},\tilde{S}')\in \textit{PredPos}|\tilde{h} \ne h\}$
\EndIf
\ElsIf {$e=m-1$}
\Comment Bits with $e<m$ are used to encode next bit counts.
\State $\tilde{S} \gets \{\}$
\Comment New positions to predict on.
\State $p \gets \max(S')$

\algstore{AS:universal_learning}
\end{algorithmic}
\end{algorithm}

\begin{algorithm}
\begin{algorithmic}[1]
\algrestore{AS:universal_learning}

\While {$|\tilde{S}|<2^m-1$}
\State $p \gets p+1$

\If {$(\exists a,b$ such that $b\in\{0,1\}$, $a^2+b=p$ and $h=\textit{triangle}(a))$ or $(h=\text{mex}(\textit{NonActiveHypotheses})$ and $\nexists a,b,\tilde{h}$ such that $b\in\{0,1\}$, $a^2+b=p$, $\tilde{h}=\textit{triangle}(a)$, $\tilde{h}\notin\textit{NonActiveHypotheses})$}
\State \Comment ``$\text{mex}(T)$'' is the minimum nonnegative integer not appearing in $T$.
\State $\tilde{S} \gets \tilde{S}\cup\{p\}$
\EndIf
\EndWhile
\State $\textit{PredPos} \gets \textit{PredPos}\cup (h,S',\tilde{S})$
\EndIf
\EndFor

\State \Comment Step 2: Predict, assuming $h$.

\State $i \gets \textit{bootstrap}(h)$
\State $S \gets \{i,i+1\}$
\State $M \gets h \text{ div } k$
\Comment $T_M$ is the machine to be simulated.
$x \text{ div } y \defeq \lfloor x/y \rfloor$.
\State $\textit{try} \gets h \bmod k$
\Comment The try number of this machine.
\State $\textit{Prediction}(i) \gets (j \text{ div } 4^\textit{try}) \bmod 2$
\State $\textit{Prediction}(i+1) \gets (j \text{ div }
(2 \cdot 4^\textit{try})) \bmod 2$

\For {$i \in 0,\ldots,n$}
\If {$\exists S,S', (h,S,S')\in\textit{PredPos}$ and $i \in S$}
\State Let $m$ be such that $2^m-1=|S'|$.
\State $e \gets |\{x\in S| x>i\}|$
\If {$e=m-1$}
\State $\textit{counter} \gets 0$
\Comment Number of $1$'s in $S'$.
\EndIf
\If {$e\ge m$}
\If {$i=n$}
\If {$\textit{Prediction}(i)=1$}
\State Accept input.
\Else
\State Reject input.
\EndIf
\EndIf
\ElsIf {$i=n$}
\State Simulate $T_M$ simultaneously on all inputs in $S'$ until
$\textit{counter}+2^e$ are accepted.
\State \Comment If this simulation does not terminate, this is a rejection
of the input.
\State Accept input.
\Else

\algstore{AS:universal_learning2}
\end{algorithmic}
\end{algorithm}

\begin{algorithm}
\begin{algorithmic}[1]
\algrestore{AS:universal_learning2}

\If {$\textit{Prediction}(i)\ne\delta_i$}
\Comment Previous simulation terminated.
\State $\textit{counter} \gets \textit{counter}+2^e$
\Comment Binary search.
\EndIf
\If {$e=0$}
\Comment $\textit{counter}$ holds the number of terminations in $S'$.
\State Simulate $T_M$ simultaneously on all inputs in $S'$ until $\textit{counter}$ are accepted.
\Comment Guaranteed to halt, if hypothesis is correct.
\State Let $\textit{Prediction}(x)$ be $0$ on all $x\in S'$ that terminated, $1$ otherwise.
\EndIf
\EndIf
\EndIf
\EndFor
\EndFunction
\end{algorithmic}
\end{algorithm}

\end{proof}

Some corollaries follow immediately.

\begin{cor}\label{C:probabilistic}
There exists a probabilistic Turing machine that is able to learn any language
in $\Pi^0_1$ with probability $1$.
\end{cor}

\begin{proof}
Instead of using a mixed strategy, it is possible to use probabilistic Turing
machines in order to generate the $m_0$ guessed bits that bootstrap each
hypothesis. In this case, there is neither a need for a mixed strategy nor a
need to consider asymptotic limits: a single probabilistic Turing machine can
perform a triangular walk over the hypotheses for $L_{\ne}$, investigating
each option an unbounded number of times. The probability that for the correct
$L_{\ne}$ at least one bootstrap guess will be correct in this way equals $1$.

The method for doing this is essentially the same as was described before.
The only caveat is that because the probabilistic TM is re-initialised at each
round and because it needs, as part of the algorithm, to simulate its actions
in all previous rounds, the TM must have a way to store its random choices,
so as to make them accessible in all later rounds.

The way to do this is to extend the hypothesis ``bootstrap'' phase from $m_0$
bits to $2 m_0$ bits. In each of the first $m_0$ bits,
the TM outputs a uniform random bit. The
$\delta_n$ bit available to it in all future rounds is then this random bit
xor the output of Player~``$\ne$''. $\delta_n$ is therefore also a uniform
random bit. In this way, in all future rounds the TM has access to these $m_0$
consistent random bits. It can then use these in the second set of $m_0$
bootstrap bits as was done with the $j$ value in the deterministic set-up.
\end{proof}

We note, as before, that the construction described continues to hold, and
therefore the results remain true, even if Oracles
are allowed, that are accessible to both players, and, in particular,
the results hold for any $\Pi^0_i$ with $i>0$:

\begin{cor}
For all $i>0$, $\Pi^0_i$ is learnable by $\Sigma^0_i$.
\end{cor}

Furthermore:

\begin{cor}\label{C:beyond}
For all $i>0$, the collection of languages learnable by $\Sigma^0_i$ is a
strict superset of $\Sigma^0_i\cup\Pi^0_i$.
\end{cor}

\begin{proof}
We have already shown that $\Sigma^0_i$ and $\Pi^0_i$ are both learnable by
$\Sigma^0_i$. Adding the $\Sigma^0_i$ languages as additional hypotheses to
Algorithm~\ref{A:universal_learning} we can see that the set
$\Sigma^0_i \cup \Pi^0_i$ is also learnable.

To give one example of a family of languages beyond this set
which is also learnable by $\Sigma^0_i$, consider the following.
Let $\Sigma^{(c)}_i$, for a fixed $c>1$, be the set of languages recognisable
by a $\Delta^0_0$ Turing machine which can make at most $c$ calls to a
$\Sigma^0_i$ Oracle.

This set contains $\Sigma^0_i$ and $\Pi^0_i$, but it also
contains, for example, the xor of any two
languages in $\Sigma^0_i$, which is outside of
$\Sigma^0_i\cup\Pi^0_i$, and therefore strictly
beyond the $i$'th level of the arithmetic hierarchy.

We will adapt Algorithm~\ref{A:universal_learning} to learn $\Sigma^{(c)}_i$.
The core of Algorithm~\ref{A:universal_learning} is its ability to use
$m$ bits of $\Delta_n$ in order to predict $2^m-1$ bits. We will, instead, use
$cm$ bits in order to predict the same amount. Specifically, we will use the
first $m$ bits in order to predict the result of the first Oracle call in each
of the predicted $2^m-1$ positions, the next $m$ bits in order to predict
the second Oracle call in each of the predicted $2^m-1$ positions, and so
on.

In total, for this to work, all we need is to replace criterion $2$ in our
list of criteria for the $\{m_t\}$ sequence with the new criterion
\[
\forall t, c m_{t+1}\le 2^{m_t}-1.
\]
An example of such a sequence is $m_t=t+\max(c,5)$.
\end{proof}

In fact, Algorithm~\ref{A:universal_learning} can be extended even beyond
what was described in the proof to Corollary~\ref{C:beyond}. For example,
instead of using a constant $c$, it is possible to adapt the algorithm to
languages that use $c(n)$ Oracle calls at the $n$'th round, for a
sufficiently low-complexity $c(n)$ by similar methods.

Altogether, it
seems that R.E.\ learning is significantly more powerful than being able to
learn merely the first level of the arithmetic hierarchy, but we do not know
whether it can learn every language in $\Delta^0_2$. Indeed, we have no
theoretical result that implies R.E.\ learning cannot be even more powerful
than the second level of the arithmetic hierarchy.

A follow-up question which may be asked at this point is whether it was
necessary to use a mixed strategy, as was used in the proof of
Theorem~\ref{T:pi_learn}, or whether a pure strategy could have been
designed to do the same.

In fact, no pure strategy would have sufficed:

\begin{lemma}
For all $i$,
\[
\inf_{L_{=}\in \Sigma^0_i} \sup_{L_{\ne}\in \textit{NA}(\Pi^0_i)} S(L_{=},L_{\ne})=1.
\]
\end{lemma}

This result is most interesting in the context of Corollary~\ref{C:probabilistic},
because it describes a concrete task that is accomplishable by a probabilistic
Turing machine but not by a deterministic Turing machine.

\begin{proof}
We devise for each $L_{=}$ a specific $L_{\ne}$ antidote. The main difficulty
in doing this is that we cannot choose, as before, $L_{\ne}=\text{co-}L_{=}$,
because $L_{\ne}$ is now restricted to be nonadaptive, whereas $L_{=}$ is
general.

However, consider $L_{\ne}$ such that its bit for round $k$ is the complement
of $L_{=}$'s response on $\Delta_{k-1}=1^{k-1}$.
This is a nonadaptive strategy, but
it ensures that $\Delta_k$ will be $1^k$ for every $k$.
Effectively, $L_{\ne}$ describes $L_{=}$'s ``red herring sequence''.
\end{proof}

\section{Approximability}\label{S:approximability}

When \emph{both} players' strategies are restricted to be nonadaptive, they
have no means
of learning each other's behaviours: determining whether their next output bit
will be $0$ or $1$ is done solely based on the present round number, not on any
previous outputs. The output of the game is therefore solely determined by
the dissimilarity of the two independently-chosen output strings.

\begin{defi}
We say that a collection of languages $\Sigma_{\ne}$ is
\textbf{approximable} by a collection of strategies $\Sigma_{=}$ if
$\text{minmax}\left(\textit{NA}(\Sigma_{=}),\textit{NA}(\Sigma_{\ne})\right)=0$.

If a collection is approximable by $\Sigma^0_1$, we simply say that it is
\textbf{approximable}.
\end{defi}

In this context it is clear that for any $\Sigma$
\[
\sup_{L_{\ne}\in\textit{NA}(\Sigma)} \inf_{L_{=}\in\textit{NA}(\Sigma)} S(L_{=},L_{\ne})=0,
\]
because $L_{=}$ can always be chosen to equal $L_{\ne}$,
but unlike in the case of adversarial learning, here
mixed strategies do make a difference.

Though we do not know exactly what the value of
$\text{minmax}\left(\textit{NA}(\Sigma^0_1)\right)$ is,
we do know the following.

\begin{lemma}\label{L:na_sup}
If $D_{=}$ and $D_{\ne}$ are
mixed strategies from $\textit{NA}(\Sigma^0_1)$, then
\begin{equation}\label{Eq:maxmin_mixed}
\sup_{D_{\ne}} \inf_{D_{=}} E\left(\limsup_{N\to\infty} \sum_{n=1}^{N} \frac{\delta_n}{N}\right) \ge \frac{1}{2}
\end{equation}
and
\begin{equation}\label{Eq:minmax_mixed}
\inf_{D_{=}} \sup_{D_{\ne}} E\left(\limsup_{N\to\infty} \sum_{n=1}^{N} \frac{\delta_n}{N}\right) \ge \frac{1}{2},
\end{equation}
where $\delta_n$ is as in the definition of the IMP game.
\end{lemma}

In other words, Player~``$\ne$'' can always at the very least break even, from a
$\limsup$ perspective.

\begin{proof}
Let $D_{\ne}$ be a mixture of the following two strategies: all zeros ($L_0$),
with probability $1/2$; all ones ($L_1$), with probability $1/2$.
By the triangle inequality, we have that for any language $L_{=}$,
\[
E\left(\limsup_{N\to\infty} \sum_{n=1}^{N} \frac{\delta_n}{N}\right)=
\frac{\text{DisSim}(L_{=},L_0)+\text{DisSim}(L_{=},L_1)}{2}\ge
\frac{\text{DisSim}(L_0,L_1)}{2}=\frac{1}{2},
\]
and because this is true for each $L_{=}$ in $D_{=}$, it is also true in
expectation over all $D_{=}$. The fact that $D_{\ne}$ is independent of
$D_{=}$ in the construction means that this bound is applicable for
both \eqref{Eq:maxmin_mixed} and \eqref{Eq:minmax_mixed}.
\end{proof}

Just as interesting (and with tighter results) is the investigation of
$\liminf$. We show
\begin{lemma}\label{L:na_inf}
\begin{equation}\label{Eq:na_inf}
\inf_{L_{=}\in\textit{NA}(\Sigma^0_1)} \sup_{L_{\ne}\in\textit{NA}(\Sigma^0_1)} \liminf_{N\to\infty} \sum_{n=1}^{N} \frac{\delta_n}{N}=
\sup_{L_{\ne}\in\textit{NA}(\Sigma^0_1)} \inf_{L_{=}\in\textit{NA}(\Sigma^0_1)} \liminf_{N\to\infty} \sum_{n=1}^{N} \frac{\delta_n}{N}=0,
\end{equation}
where $\delta_n$ is as in the definition of the IMP game.
\end{lemma}

\begin{proof}
Let $\textit{triangle}(x)$ be as in \eqref{Eq:triangle}, and
let $\textit{caf}(x)$ be the maximum integer, $y$, such that $y!\le x$.

The language $L_{=}$ will be defined by
\[
w_i \in L_{=} \Leftrightarrow w_i \in L_{\textit{triangle}(\textit{caf}(i))},
\]
where $L_0,L_1,L_2,\ldots$
is an enumeration over all R.E.\ languages.

To prove that for any $j$, if $L_{\ne}=L_j$ the claim holds, let us first join
the rounds into ``super-rounds'', this being the partition of the rounds set
according to the value of $y=\textit{caf}(i)$. At each super-round, $L_{=}$
equals a specific $L_x$,
and by the end of the super-round, a total of $(y-1)/y$ of the
total rounds will have been rounds in which $L_{=}$ equals this
$L_x$. Hence, the Hamming distance between the two (the number of differences)
at this time is at most $1/y$ of the string lengths. Because
each choice of $x$ repeats an infinite number of times, the $\liminf$ of this
proportion is $0$.
\end{proof}

With this lemma, we can now prove Theorem~\ref{T:approx}.

\begin{proof}
The theorem is a direct corollary of the proof of Lemma~\ref{L:na_inf}, because
the complement of the language $L_{=}$ that was constructed in the proof to
attain the infimum can be used as $\bar{L}$.
\end{proof}

Combining Lemma~\ref{L:na_sup} and \ref{L:na_inf} with the definition of the
payoff function in \eqref{Eq:ne}, we get, in total:
\begin{cor}
\[
1/4 \le \text{maxmin}\left(\textit{NA}(\Sigma^0_1)\right)\le 1/2
\]
and
\[
1/4 \le \text{minmax}\left(\textit{NA}(\Sigma^0_1)\right)\le 1/2.
\]
\end{cor}

Though we have the exact value of neither maxmin nor minmax in this case,
we do see that the case is somewhat unusual in that neither player has a
decisive advantage.

\section{Conclusions and further research}

We have introduced the IMP game as an arena within which to test the ability
of algorithms to learn and be learnt, and specifically investigated three
scenarios:
\begin{description}
\item[Adversarial learning,] where both algorithms are simultaneously trying to
learn each other by observations.
\item[Non-adversarial (conventional) learning,] where an algorithm is
trying to learn a language by examples.
\item[Approximation,] where languages (or language distributions) try to mimic
each other without having any visibility to their opponent's actions.
\end{description}

In the case of adversarial learning, we have shown that $\Sigma^0_i$ can learn
$\Sigma^0_i$
but not $\Pi^0_i$.

In conventional learning, however, we have shown that $\Sigma_i$ can
learn $\Sigma^0_i$, $\Pi^0_i$ and beyond into the
$(i+1)^{\rm th}$
level of the
arithmetic hierarchy, but this learnability is yet to be upper-bounded.
Our conjecture
is that the class of learnable languages is strictly a subset of
$\Delta^0_2$. If so, then this defines a new class of languages between the
first and second levels of the arithmetic hierarchy, and, indeed, between
any consecutive levels of it.

Regarding approximability, we have shown that (unlike in the previous
results) no side has the absolute upper hand in the game, with the game value
for Player~``$\ne$'', if it exists, lying somewhere between $1/4$ and $1/2$.
We do not know, however, whether the game is completely unbiased or not.

An investigation of adversarial learning in the context of recursive
languages was given as a demonstration of the fact that in IMP it may be the
case that no Nash equilibrium exists at all, and pure-strategy learning was
given as a concrete example of a task where probabilistic Turing machines have
a provable advantage over deterministic ones.

\bibliographystyle{plain}

\begin{thebibliography}{10}

\bibitem{Dowe2008a}
D.L. Dowe.
\newblock {Foreword re C.\ S.\ Wallace}.
\newblock {\em {Computer Journal}}, 51(5):523--560, September 2008.
\newblock Christopher Stewart WALLACE (1933-2004) memorial special issue.

\bibitem{Dowe2008b}
D.L. Dowe.
\newblock {Minimum Message Length and statistically consistent invariant
  (objective?) Bayesian probabilistic inference -- from (medical)
  ``evidence''}.
\newblock {\em Social Epistemology}, 22(4):433--460, Oct--Dec 2008.

\bibitem{Dowe2011a}
D.L. Dowe.
\newblock {MML, hybrid Bayesian network graphical models, statistical
  consistency, invariance and uniqueness}.
\newblock In {Bandyopadhyay, P.S. and Forster, M.R.}, editor, {\em {Handbook of
  the Philosophy of Science -- Volume 7: Philosophy of Statistics}}, pages
  901--982. {Elsevier}, 2011.

\bibitem{Dowe2013a}
D.L. Dowe.
\newblock {Introduction to Ray Solomonoff 85th Memorial Conference}.
\newblock In {\em Proceedings of Solomonoff 85th memorial conference -- Lecture
  Notes in Artificial Intelligence (LNAI)}, volume 7070, pages 1--36. Springer,
  2013.

\bibitem{DoweHernandez-OralloDas2011}
D.L. Dowe, J.~Hern{\'a}ndez-Orallo, and P.K. Das.
\newblock Compression and intelligence: Social environments and communication.
\newblock In {\em AGI: 4th Conference on Artificial General Intelligence --
  Lecture Notes in Artificial Intelligence (LNAI)}, pages 204--211, 2011.

\bibitem{flake1998computational}
G.W. Flake.
\newblock {\em The Computational Beauty of Nature: Computer Explorations of
  Fractals, Chaos, Complex Systems, and Adaptation}.
\newblock A Bradford book. Cambridge, Massachusetts, 1998.

\bibitem{Gold1967}
E.M. Gold.
\newblock Language identification in the limit.
\newblock {\em Information and Control}, 10(5):447--474, 1967.

\bibitem{Hernandez-OralloDoweEspana-CubilloHernandez-LloredaInsa-Cabrera2011}
J.~Hern{\'a}ndez-Orallo, D.L. Dowe, S.~Espa{\~n}a-Cubillo, M.V.
  Hern{\'a}ndez-Lloreda, and J.~Insa-Cabrera.
\newblock On more realistic environment distributions for defining, evaluating
  and developing intelligence.
\newblock In {\em AGI: 4th Conference on Artificial General Intelligence --
  Lecture Notes in Artificial Intelligence (LNAI)}, volume 6830, pages 82--91.
  Springer, 2011.

\bibitem{Huang:adversarial}
Ling Huang, Anthony~D. Joseph, Blaine Nelson, Benjamin~I.P. Rubinstein, and
  J.~D. Tygar.
\newblock Adversarial machine learning.
\newblock In {\em Proceedings of the 4th ACM Workshop on Security and
  Artificial Intelligence}, AISec '11, pages 43--58, New York, NY, USA, 2011.
  ACM.

\bibitem{LewisRichardson1966}
D.K. Lewis and J.S. Richardson.
\newblock Scriven on human unpredictability.
\newblock {\em Philosophical Studies: An International Journal for Philosophy
  in the Analytic Tradition}, 17(5):69--74, October 1966.

\bibitem{Wei:adversarial}
Wei Liu and Sanjay Chawla.
\newblock {A Game Theoretical Model for Adversarial Learning}.
\newblock In {Saygin, Y and Yu, JX and Kargupta, H and Wang, W and Ranka, S and
  Yu, PS and Wu, XD}, editor, {\em {2009 IEEE INTERNATIONAL CONFERENCE ON DATA
  MINING WORKSHOPS (ICDMW 2009)}}, pages {25--30}. {Knime; Mitre; CRC Press},
  {2009}.
\newblock {9th IEEE International Conference on Data Mining, Miami Beach, FL,
  DEC 06-09, 2009}.

\bibitem{Lowd:adversarial}
Daniel Lowd and Christopher Meek.
\newblock Adversarial learning.
\newblock In {\em Proceedings of the Eleventh ACM SIGKDD International
  Conference on Knowledge Discovery in Data Mining}, KDD '05, pages 641--647,
  New York, NY, USA, 2005. ACM.

\bibitem{Nash1951:Non-cooperative_Games}
J.~Nash.
\newblock {Non-cooperative Games}.
\newblock {\em The Annals of Mathematics}, 54(2):286--295, 1951.

\bibitem{vonNeumann:1944:TGE}
J.v. Neumann and O.~Morgenstern.
\newblock {\em Theory of Games and Economic Behavior}.
\newblock Princeton University Press, Princeton, NJ, 1944.

\bibitem{Rogers:recursive}
Hartley Rogers, Jr.
\newblock {\em Theory of Recursive Functions and Effective Computability}.
\newblock MIT Press, Cambridge, MA, second edition, 1987.

\bibitem{Scriven1965}
M.~Scriven.
\newblock {An essential unpredictability in human behavior}.
\newblock In B.B. Wolman and E.~Nagel, editors, {\em Scientific Psychology:
  Principles and Approaches}, pages 411--425. Basic Books (Perseus Books),
  1965.

\bibitem{Solomonoff1978}
R.J. Solomonoff.
\newblock Complexity-based induction systems: Comparisons and convergence
  theorems.
\newblock {\em IEEE Transaction on Information Theory}, {IT-24}(4):422--432,
  1978.

\bibitem{Solomonoff2009}
R.J. Solomonoff.
\newblock Algorithmic probability: Theory and applications.
\newblock In F.~Emmert-Streib and M.~Dehmer, editors, {\em Information Theory
  and Statistical Learning, Springer Science and Business Media}, pages 1--23.
  Springer, N.Y., U.S.A., 2009.

\bibitem{Solomonoff2010}
R.J. Solomonoff.
\newblock Algorithmic probability, heuristic programming and {AGI}.
\newblock In {\em Proceedings of the Third Conference on Artificial General
  Intelligence, AGI 2010}, pages 251--257, Lugano, Switzerland, March 2010.
  IDSIA.

\bibitem{Solomonoff2011}
R.J. Solomonoff.
\newblock Algorithmic probability -- its discovery -- its properties and
  application to strong {AI}.
\newblock In H.~Zenil, editor, {\em Randomness Through Computation: Some
  Answers, More Questions}, pages 1--23. World Scientific Publishing Co., Inc.,
  River Edge, NJ, USA, 2011.

\bibitem{Turing:Computable}
A.M. Turing.
\newblock On computable numbers, with an application to the
  {E}ntscheidungsproblem.
\newblock {\em Proc. London Math. Soc.}, 42:230--265, 1936.

\bibitem{Valiant1984}
Leslie~G Valiant.
\newblock A theory of the learnable.
\newblock {\em Communications of the ACM}, 27(11):1134--1142, 1984.

\end{thebibliography}

\end{document}